\documentclass{article}[11pt]

\usepackage{graphicx}
\usepackage{amssymb}
\usepackage{amsthm}
\usepackage{amsmath}
\usepackage{amsfonts}
\usepackage{fullpage}
\usepackage[font=small,labelfont=bf]{caption}

\usepackage{subfigure}

\newcommand{\frechet}{Fr\'echet}
\newcommand{\dfre}{d_F}
\newcommand{\dhaus}{d_H}

\newcommand{\dist}{\mathit{d}}

\newcommand{\BO}{\mathcal{O}}

\newcommand{\dfd}{discrete \frechet\ distance}
\newcommand{\npc}{\textbf{NP}-complete}

\newcommand{\scm}{set-chain}
\newcommand{\SCM}{Set-Chain}
\newcommand{\true}{\textbf{true}}
\newcommand{\false}{\textbf{false}}

\newtheorem{definition}{Definition}
\newtheorem{theorem}{Theorem}
\newtheorem{corollary}{Corollary}

\begin{document}

\title{Intermittent Map Matching with the Discrete \frechet\ Distance}


\author{Tim Wylie\thanks{The University of Texas-Rio Grande Valley, Edinburg, TX 78539  {\tt timothy.wylie@utrgv.edu}} \and Binhai~Zhu\thanks{Montana State University, Bozeman, MT 59717-3880  {\tt bhz@cs.montana.edu}}}


\date{}
\maketitle
\begin{abstract}
In this paper we focus on the map matching problem  
where the goal is to find a path through a planar graph such that the path through
the vertices closely matches a given polygonal curve.
The map matching problem is usually approached with the \frechet\ distance matching
the edges of the path as well.  Here, we formally define the discrete map matching problem 
based on the \dfd.  We then look at the complexities of
some variations of the problem which allow for vertices in the graph to be unique or
reused, and whether there is a bound on the length of the path or the number
of vertices from the graph used in the path.
We prove several of these problems to be \npc, and then conclude the paper with some open
questions.
\end{abstract}
\section{Introduction}

The map matching problem arose naturally out of sophisticated 
GIS (Geographic Information System) software and algorithms.  With the
rise of excellent satellite data and the common use of GPS (Global Positioning System) 
in cellphones and vehicles for navigation, an important task is map routing.
There are several problems dealing with map routing, but our focus will be on the 
reconstruction tasks.  Assume that we know the road network and can treat it
as a planar embedded graph, and that we also have the GPS data from a car
traveling.  Often, this data is not accurate and may be noisy or approximate to
the exact location of the vehicle.  Thus, the data may not align with the roads
properly.

The standard map matching problem is trying to find the most likely route of the vehicle on 
the road network given a planar metric graph to represent the road network and a polygonal curve 
representing the GPS data.
We extend this analogy for modern instances where the data may be intermittent due to 
coverage or power constraints.  A common issue with cellphones or hand-held GPS devices is 
limited battery life.  A feasible situation is one where a user may only power on the 
device when near a city or when they need something immediately (such as making a call 
for directions), and afterwards they turn the device back off.  
Similarly, in remote areas, a user may only have reception near certain towns.  
Thus, tracking their phone is only possible when a signal is available.
In these instances our data has discrete points of the lossy data.
We still have a polygonal curve,
but we can not depend on all edges of the line to be accurate location data.
We know the node ordering, but the edges may not represent the actual path taken.
So the problem is to find the most probable \emph{simple} path of a vehicle 
between the recorded points on the road network.

With respect to map matching, the problem of finding a path in a graph given a polygonal 
curve with respect to the \frechet\ distance was first posed by Alt et al. \cite{Alt:2003:JALGS}
as follows: Let $G = (V, E)$ be an undirected connected planar graph with a given straight-line
embedding in $\mathbb{R}^2$ and a polygonal line $P$. Find a path $\pi$ in $G$ which minimizes
the \frechet\ distance between $P$ and $\pi$. They give an efficient algorithm which runs in
$\BO(pq \log q)$ time and $\BO(pq)$ space where $p$ is the number of line segments of $P$ and $q$ is the
complexity of $G$. Their version allows for vertices and edges to be visited multiple times,
and is similar to the discrete version we cover in Section \ref{sec:nmmck}.

The recent work by Maheshwari et al. improved the running time for the 
map matching problem for complete graphs \cite{Maheshwari:2011:CCCG}.
The original algorithm would decide it in $\BO(pk^2 \log k)$ where $k$ is the number of vertices
in the graph, and their new algorithm solves it in $\BO(pk^2)$.  
We refer to this as the set-chain matching problem as in \cite{Wylie:2014:TCS}. 

Map matching is an active area of research with many approaches.  The two main 
methodologies are those based on geometric methods and ones based on Global Weight Optimization.
However, the methodologies can also be classified based on the problem definition where
we have local/incremental methods, global methods, and statistical methods.
These can all be extended to include topological and geological conditions, 
current weather and traffic conditions, speed limits, and other variables that can 
produce more optimal routes \cite{Lou:2009:GIS,Wei:2013:TECH}. 
However, the work presented here is fully based on global geometric methods. 
We assume that we have all of the data and want to use a geometric method (the \dfd ) 
to find the best fit for the input.

There has been recent work which allowed for better performance with certain types
of curves, with dual simplification for an approximate result, with bounded simplification
of one of the chains, and in graphs with certain properties, \cite{Brakatsoulas:2005:VLDB,Buchin:2009:SODA,Chen:2011:ALENEX,Driemel:2012:DCG}.
With map matching, for the weak \frechet\ distance, the bounds have been lowered further to
$\BO(pq)$ \cite{Chen:2008:UP}, and the problems can be better defined with a smaller error bound
\cite{Wenk:2006:SSDBM}.

In reality, all GPS data is discrete, and these approaches smooth the data.  There are some
methods optimized for low-sampling-rate data \cite{Lou:2009:GIS}, but even these assume
some maximum time between samples (less than five minutes).  Our purpose is to analyze 
data where samples may be hours apart and can not be reasonably `smoothed'. 

This paper is organized as follows. In Section \ref{sec:preliminaries} we cover some necessary
preliminary concepts. Section \ref{sec:dismmatch} formally defines the problem and scope of
our investigation. Then Sections \ref{sec:nmmck}, \ref{sec:nmmsk}, and \ref{sec:ummk} cover the
complexities of the three variations we are considering.  
Finally, we conclude the paper in \ref{sec:conclusion} with some future research directions and open questions.


\section{Preliminaries} \label{sec:preliminaries}

The discrete \frechet\ distance was originally defined by Eiter and Mannila \cite{Eiter:1994:TECH} 
in 1994, and was further expanded on theoretically by Mosig et al. in 2005 \cite{Mosig:2005:CGTA}.

Given two polygonal curves, we define the discrete \frechet\ distance as follows.
We use $\dist(a,b)$ to represent the
euclidean distance between two points $a$ and $b$, but it could be
replaced with other distance measures depending on the application. 

\begin{definition} \label{def:dfdformal}
    The discrete Fr\'echet distance, $\dfre$, between two polygonal curves $f:[0,m] \rightarrow \mathbb{R}^k$ and $g:[0,n] \rightarrow \mathbb{R}^k$ is defined as:
    \[
    \dfre(f,g) = \min_{\sigma:[1:m+n]\rightarrow[0:m], \beta:[1:m+n]\rightarrow[0:n]}\max_{s \in [1:m+n]} \Bigg \{d \Big ( f(\sigma(s)), \, g(\beta(s)) \Big ) \Bigg \}
    \]
    where $\sigma$ and $\beta$ range over all discrete non-decreasing onto mappings of the form $\sigma:[1:m+n]\rightarrow[0:m], \beta:[1:m+n]\rightarrow[0:n]$.
\end{definition}

The continuous \frechet\ distance is typically explained as the relationship between a 
person and a dog connected by a leash walking along the two curves and trying
to keep the leash as short as possible.  However, for the discrete case, we only
consider the nodes of these curves, and thus the man and dog must ``hop'' along
the nodes.
With a dynamic programming solution for finding the discrete \frechet\ distance between 
two polygonal curves, Eiter and Mannila proved:

Since the moves taken along the chains are discrete, finding the best walk between the two chains
is relatively straightforward.
By giving a dynamic programming solution for finding the discrete \frechet\ distance between 
two polygonal curves, in \cite{Eiter:1994:TECH} Eiter and Mannila proved that it could be easily solved in $\BO(mn)$ time.
Recently, Agarwal et al. provided the first subquadratic algorithm for the \dfd\ giving the following theorem.

\begin{theorem} \label{thm:dfdtime}
The discrete \frechet\ distance between two polygonal curves, with $m$ and $n$ vertices 
respectively, can be computed in $\BO(\frac{mn\log\log n}{\log n})$ time \cite{Agarwal:2013:SODA}.
\end{theorem}

\begin{figure}[ht!]
    \begin{center}
        \subfigure[]{\label{fig:frecdiffa}\includegraphics[height=.65in]{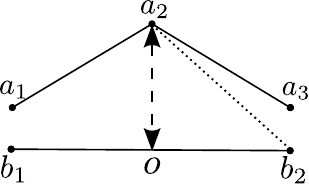}}
        \hspace*{.5cm}
        \subfigure[]{\label{fig:frecdiffb}\includegraphics[height=.65in]{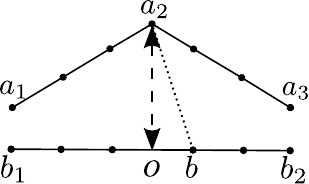}}
        \hspace*{.5cm}
        \subfigure[]{\label{fig:frec_haus_diff}\includegraphics[height=.6in]{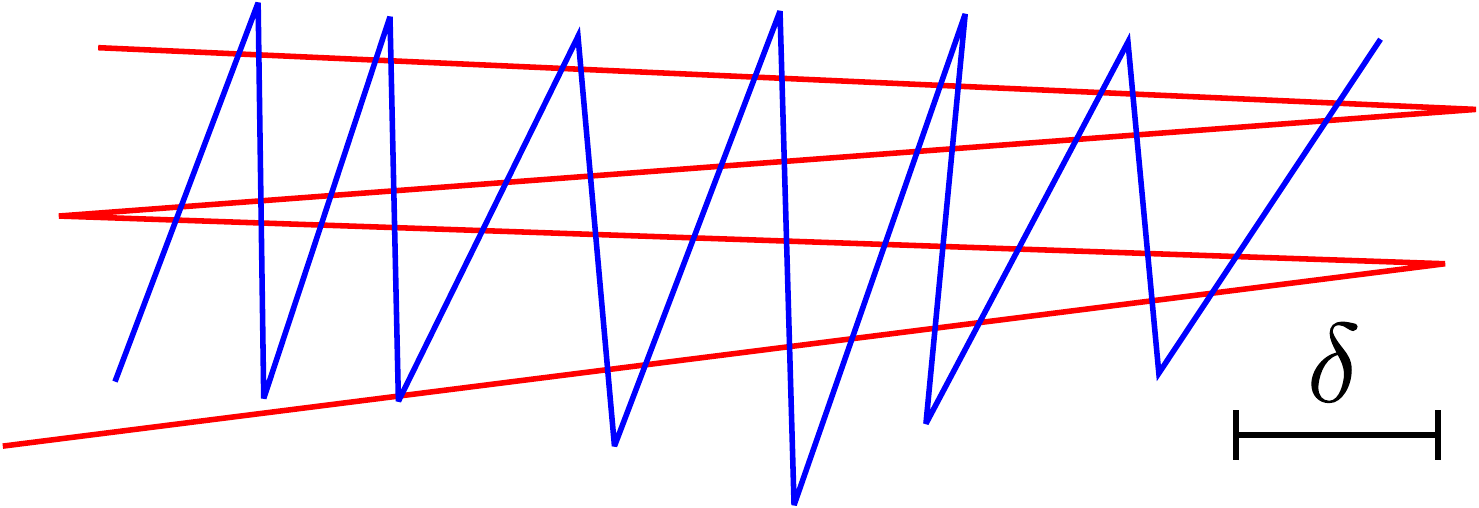}} \\

    \end{center}
    \caption{Figures (a) and (b) show the relationship between the discrete and continuous \frechet\ distance where $o$ is the 
    continuous and the dotted line represents the discrete distance. (a) the curves
    have fewer nodes and a larger discrete \frechet\ distance, while (b) is the same paths with
    more nodes, and thus provides a better approximation of the \frechet\ distance.
    (c) Two polygonal curves that highlight the difference between the \frechet\ distance and the 
        Hausdorff distance. Here $\dhaus \leq \delta$, but $\dfre > 2\delta$.}
    \label{fig:frec_diff}
\end{figure}

Figure \ref{fig:frec_diff} shows the relationship between the discrete and continuous
\frechet\ distances. In Figure \ref{fig:frecdiffa}, we have two polygonal curves (or chains) 
$\langle a_1, a_2, a_3 \rangle$ and $\langle b_1, b_2 \rangle$, 
the continuous \frechet\ distance between the two is the distance
from $a_2$ to segment $\overline{b_1b_2}$, i.e., $\dist(a_2,o)$. The discrete \frechet\
distance is $\dist(a_2,b_2)$. The discrete \frechet\ distance could be
quite larger than the continuous distance.
However, with enough evenly sample points on the two curves, the resulting
discrete \frechet\ distance, i.e., $\dist(a_2,b)$ in Figure \ref{fig:frecdiffb}, 
closely approximates $\dist(a_2,o)$.

The Hausdorff distance was first defined by Felix Hausdorff in 1914 \cite{Hausdorff:1914:BOOK}.
Since its introduction, the Hausdorff distance has become one of the most widely used 
similarity measures across many disciplines.

\begin{definition} \label{def:hausdorff}
    Let $X$ and $Y$ be two non-empty subsets of a metric space $(M,d)$ where $M$ is the space and $d$ the distance measure. 
    We define their Hausdorff distance $\dhaus(X, Y)$ by
    \[
     \dhaus(X,Y) = \max \{\sup_{x \in X} \inf_{y \in Y} d(x,y), \sup_{y \in Y} \inf_{x \in X} \dist(x,y) \},
    \]
    where $\sup$ represents the supremum and $\inf$ the infimum.
\end{definition}

Figure \ref{fig:frec_haus_diff}
shows a classic example used to demonstrate this difference.  There are two polygonal 
curves which intersect repeatedly.  Due to this crossing the Hausdorff distance is
less than or equal to $\delta$, however, because the curves zigzag in different directions
the \frechet\ distance is greater than $2\delta$ \cite{Alt:1995:JCOMPS}.  


Recently, we worked on some similar variations for the discrete set-chain matching problem.  The difference
is that the objective of set-chain matching is to find a path through a set of points rather than a graph.  
We reduce from one of these problems so we briefly give the definition and the results.

\begin{definition}[The Discrete \SCM\ Matching Problem]\hfill \\
    \noindent
    {\bf Instance:}
    Given a point set $S$, a polygonal curve $P$ in $\mathbb{R}^d$ $(d \geq 2)$, an integer $K \in \mathbb{Z}^+$,  
    and an $\varepsilon > 0$. \\
    {\bf Problem:}
    Does there exist a polygonal curve $Q$ with vertices chosen from $S'$ 
    where $S' \subseteq S$, 
    such that $T \leq K$ and $\dfre(P,Q) \leq \varepsilon$?
\end{definition}

$T$ is defined in two ways.  When limiting the number of nodes in the curve, $T=|Q|$, and 
if restricting the number of points used then $T=|S'|$.
They vary whether 
there is a uniqueness constraint on $s \in S$ being used as a node in $Q$ (if points may be used more than once),
and whether our goal is to limit the size of the curve $Q$ or the set $S'$.  We distinguish the problems as
Unique/Non-unique(U/N) \SCM(S) Matching(M) with a $K$ Subset/Curve(S/C).  
The variants are thus NSMS, NSMC, and USM.

\begin{theorem} \label{thm:nsmck}
    The discrete non-unique \scm\ matching problem where $T=|Q|$ is polynomial, i.e., NSMC $\in$ \textbf{P} \cite{Wylie:2014:TCS}.
\end{theorem}

\begin{theorem} \label{thm:nsmsk}
    The discrete non-unique \scm\ matching (NSMS) problem where $T=|S'|$ is \npc\ \cite{Wylie:2014:TCS}.
\end{theorem}

\begin{theorem} \label{thm:usmck}
    The discrete unique \scm\ matching (USM) problem is \npc\ \cite{Wylie:2014:TCS}.
\end{theorem}


\section{Discrete Map Matching} \label{sec:dismmatch}

The definition of discrete map matching is similar to the set-chain matching definition 
in \cite{Wylie:2014:TCS}, and has three variants that we will consider.  

\begin{definition}[The Discrete Map Matching Problem]\hfill \\
    \noindent
    {\bf Instance:}
    Given a simple connected planar graph $G=(V,E)$ embedded in $\mathbb{R}^2$, 
    a polygonal curve $P$ in $\mathbb{R}^d$ $(d \geq 2)$, an integer $K \in \mathbb{Z}^+$,  
    and an $\varepsilon > 0$. \\
    \noindent
    {\bf Problem:}
    Does there exist a path $Q$ in $G$ with the polygonal curve using
    vertices chosen from $V'$ where $V' \subseteq V$, 
    such that $T \leq K$ and $\dfre(P,Q) \leq \varepsilon$?
\end{definition}

$T$ is defined in two ways.  When we are minimizing the size of the chain, $T=|Q|$.
If we are minimizing the vertices in the graph used then $T=|V'|$.   
We look at the analogous
versions of the set-chain matching problems for each of these: NMMC and NMMS.
We then consider the version where the vertices in the path must be unique and only
used once, which we label UMM.  
Note that when the vertices are unique the two minimization problems ($|Q|$,$|V'|$)
are equivalent.
For reference, the naming convention is Unique/Non-unique(U/N) 
Map(M) Matching(M) with a $K$ Subset/Chain(S/C).  





\section{Map Matching with $T=|Q|$ (NMMC)}\label{sec:nmmck}


When focused solely on the length of $Q$, the problem is similar to the set-chain
version (NSMC) \cite{Wylie:2014:TCS}.  The problem has an almost identical 
optimal substructure, so we forgo the proof here.  The recurrence to find the minimum
size of $Q$ (in number of vertices) is given in Equation \ref{eq:nmmck}. The
actual dynamic programming algorithm is also omitted, but is straightforward.
The recurrence uses a 2D array $M$ of size $|V| \times |P|$ where the first row and column are
initialized to one if $\dist(v_1,p_k) \leq \varepsilon$ where $1 \leq k \leq |P|$ 
and $\dist(v_k,p_1) \leq \varepsilon$ where $1 \leq k \leq |V|$. The values are set to $K+1$ otherwise.
This is polynomial with the worst case being
in a complete graph, which is equivalent to the discrete set-chain matching version.  
If the graph is not complete, the complexity will be lower since each vertex $v$ only looks at 
its neighbor set, $N(v)$.


\begin{equation}
    M[i,j] = \min
    \begin{cases}
        M[i,j\mbox{-}1], & \hspace*{-.2cm} \mbox{if } \dist(v_i,p_j) \leq \varepsilon, v_i \in N(v_{i\mbox{-}1}), 
                 M[i,j\mbox{-}1] \neq K\mbox{+}1 \\
        \underset{k \in N(v_i)}{\min} \hspace*{-.1cm} M[k,j\mbox{-}1] + 1, & \hspace*{-.2cm} \mbox{if } \dist(v_i,p_j) \leq \varepsilon, v_i \in N(v_{i\mbox{-}1}), 
                 M[i,j\mbox{-}1] = K\mbox{+}1 \\
        K+1, & \hspace*{-.2cm} \mbox{if } \dist(v_i,p_j) > \varepsilon \mbox{ or } v_i \notin N(v_{i\mbox{-}1})
    \end{cases}
    \label{eq:nmmck}
\end{equation}


\begin{theorem}
    The discrete Non-unique Map Matching (NMMC) problem where $T=|Q|$ is in \textbf{P}.
\end{theorem}

\begin{proof}
    This problem is similar to NSMC (\cite{Wylie:2014:TCS}) with the restriction of 
    which vertices are viable given the previous choice.  
    Rather than looking at all values for the last
    vertex, it is restricted to only those vertices which have an edge
    between them.  
    Thus, the problem has a similar solution and is polynomial.
    \hfill $\square$
\end{proof}


\section{Map Matching with $T=|V'|$ (NMMS)} \label{sec:nmmsk}

This problem has proven more difficult to analyze than the other versions
of the discrete map matching problem.  
With the \dfd\ we show that it is \npc\ for general graphs (Theorem \ref{thm:nmmskgen}), 
and we show that with planar graphs the problem is \npc\ under the
Hausdorff distance (Theorem \ref{thm:nmmskhaus}). Even though we show these results,
we did not prove the complexity for the planar version under the \dfd . 
We believe this problem to also be \npc, however, we leave this problem for future work.

\begin{corollary} \label{thm:nmmskgen}
    Discrete non-unique map matching where $T=|V'|$ in general graphs is \npc .
\end{corollary}

\begin{proof}
    By a simple reduction from NSMS (Theorem \ref{thm:nsmsk}) we show this is true.
    Given a set of points $S$, a polygonal curve $P$, $\varepsilon > 0$, and a $K \in \mathbb{Z}^+$, 
    we build a complete graph $G=\{V,E\}$ with
    the vertices being the points in $S$, i.e. $V=S$ and the number of edges $|E| = {|S| \choose 2}$. 
    
    This embedding allows all possible paths to be
    explored for $Q$, and the path it returns will be the same for both NMMS and NSMS.
    Thus, there exists a polygonal curve $Q$ with nodes taken from $S' \subset S$, such that
    $|S'| \leq K$ and $\dfre(P,Q) \leq \varepsilon$ if and only if there exists a path $Q'$
    in $G$ such that the vertices in $Q'$ are taken from $V' \subset V$, such that
    $|V'| \leq K$ and $\dfre(P,Q') \leq \varepsilon$.
    \hfill $\square$
\end{proof}

\begin{theorem} \label{thm:nmmskhaus}
    Discrete non-unique map matching with $T=|V'|$ in planar graphs under the Hausdorff distance is \npc .
\end{theorem}


\begin{proof}
    This is a straight-forward reduction from the Hamiltonian circuit problem in grid graphs \cite{Itai:1982:SIAM}.
    Let $G_{k}$ be a grid graph with $k$ vertices such that any vertex $v \in G_k$ is located at 
    $v_x,v_y \in \mathbb{Z}^+$ and $G_{k}$ is an induced 
    subgraph of the infinite unit grid graph $G_{\infty}$.
    For our construction we let $G = G_{k}$ and we let $P$ be the vertices $V$ of $G$ ordered in any way as
    a polygonal chain.
    
    There exists a Hamiltonian circuit in $G_{k}$ if and only if 
    there exists a path $Q$ in $G$ such that $|Q|=|V|$ and $\dhaus(P,Q) = 0$.

    If there is a Hamiltonian circuit $Q'$ on the graph, then by definition, this path 
    (with the start/end vertex arbitrarily chosen) has $\dhaus(G,Q') = 0$
    since it covers every edge and vertex.  Also, since $|Q'|=|V|$, $|Q|=|V|$. 
    
    Given a path $Q$ in $G=G_{k}$ where $|Q|=|V|$ and $\dhaus(P,Q) = 0$. Note that
    this path must visit every edge and vertex since $\dhaus(P,Q) = 0$.  If $|Q|=|V|$,
    it must only visit each vertex once, and thus the path must be a Hamiltonian
    circuit.
    \hfill $\square$
\end{proof}



\section{Unique Map Matching (UMM)} \label{sec:ummk}


Map matching with unique vertices is an interesting and relevant problem.
In most applications related to map matching where the graph is 
planar, rarely would a vertex be visited multiple times.  In the 
GPS application of a vehicle on a road network, this would be equivalent to a car 
visiting the same intersection multiple times.  This may occur, but when trying to 
find the likely path of a vehicle from the origin to the destination we can
disregard self-intersecting data as unimportant overall.

We address discrete unique map matching where any vertex in the graph can be used
at most once in the path, and show that this problem is \npc\ 
via a reduction from planar 3-SAT \cite{Lichtenstein:1982:SIAM}. Planar 3-SAT is
a 3-SAT formula that can be drawn as a planar graph with vertices representing clauses
and variables, and the edges representing inclusion of a variable in a clause.  
This is a convenient form of 3-SAT for geometric reductions since 
a crossover gadget is unnecessary.

By standard convention, we first introduce several
planar ``gadgets'' that we then arrange in our reduction.  We will build up the gadgets in
a piecewise manner, and then show how they are connected to form a single polygonal curve and a planar graph.
Due to the length of this section, we cover the gadgets and then formally do the
reduction with the assumption of their correctness.
  
Let $\varphi$ be the 3-SAT formula represented by the input instance of planar 3-SAT with $N$ variables 
and $M$ clauses. 
Given an $\varepsilon > 0$, we construct a planar graph $G$ and a polygonal curve $P$.
We show that $\varphi$ is satisfiable if and only if with our construction there exists a path 
$Q$ with unique nodes from the vertices in $G$ such that $\dfre(P,Q) \leq \varepsilon$.


\subsection{The Chain Gadget}
In order to retain a \true\ or \false\ selection, we first show a `chain' gadget used to transfer information
from the variables to the clauses. 
Figures \ref{fig:umm_chain_true} and \ref{fig:umm_chain_false} show a polygonal curve with a ladder graph structure
that constitute a chain with \true\ and \false\ paths shown, respectively.  
Since the nodes used in $G$ may only be used once, a path through this graph that maintains
$\dfre(P,Q) \leq \varepsilon$ only has two possibilities.  Starting at the top
left vertex a path going down is a \true\ setting and going to the right is \false.

\begin{figure}[ht!]
    \begin{center}
        \subfigure[True Chain]{\label{fig:umm_chain_true}\includegraphics[width=.37\textwidth]{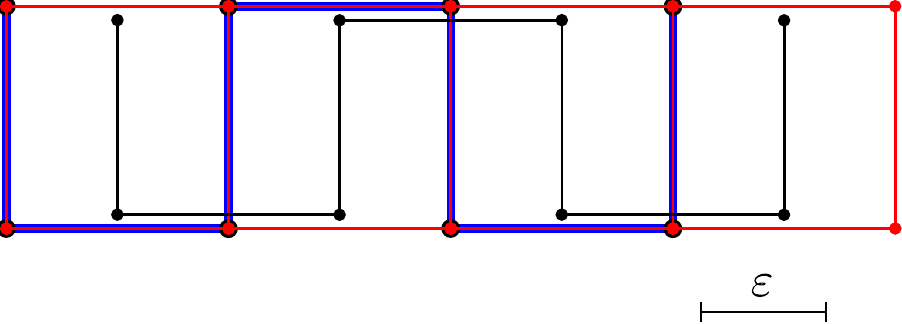}}
        \hspace*{.2cm}
        \subfigure[False Chain]{\label{fig:umm_chain_false}\includegraphics[width=.37\textwidth]{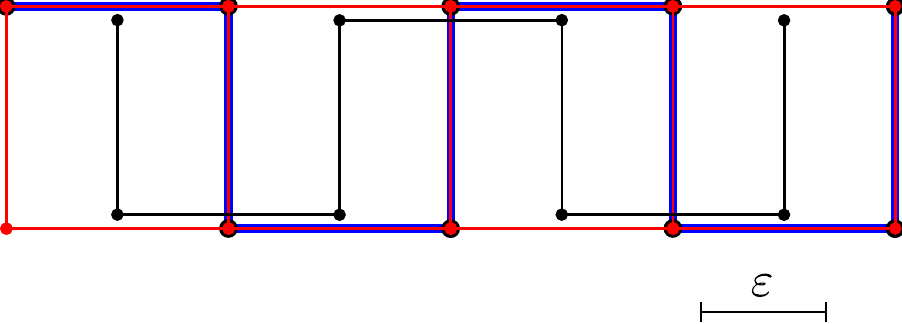}}
        \hspace*{.2cm}
        \subfigure[Elbow]{\label{fig:umm_elbow}\includegraphics[width=.18\textwidth]{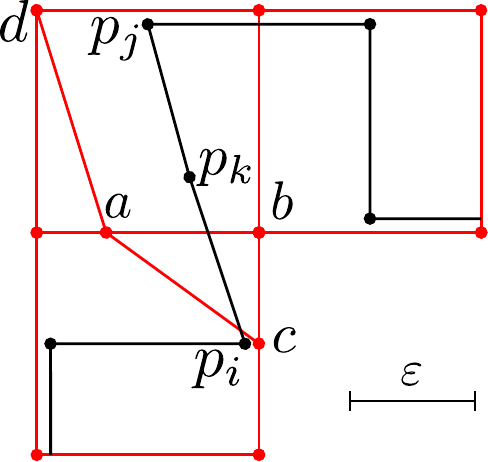}} \\
    \end{center}
    \caption{(a) A chain with a \true\ path. (b) A chain with a \false\ path. (c) An elbow in order for the chain to make turns.}
    \label{fig:umm_temp}
\end{figure}

Figure \ref{fig:umm_elbow} shows how the chain can make a right angled turn 
while maintaining the \true\ or \false\ path.  This allows the chains to be configured in a natural way with a 
clean layout. Note the extra edges of the graph ($e_{ac}$ and $e_{ad}$) in the corner which allows
the \true\ path to go around the vertex $b$. Since a \false\ path must pass through
$a$ and $b$ to cover $p_i$ and then $p_{k}$, it must continue up from $b$ in order to cover $p_j$. 
The \true\ path goes through $c$ and has already covered $p_i$ so it can then go 
through $a$ via $e_{ac}$ and cover $p_k$, and then via $e_{ad}$ it goes around the
outside of the corner to cover $p_j$ and then comes down to $b$. 



\subsection{The Variable Gadget}
A chain is the basis of the variable gadget, but we also attach an additional 
diamond structure along with another polygonal curve for the diamond graph components.
Figure \ref{fig:umm_var} shows the planar graph and the two pieces of $P$ 
necessary for each variable gadget.   
The additional diamond shapes are needed in order to force the
alternation between \true\ and \false\ states.  

Setting a variable \true\ or \false\ works identically to the chain gadgets. 
For a variable $x_i$, to set the variable \true, the path begins at $c_{i_1}$ and 
visits vertex $d_{i_1}$ next (Figure \ref{fig:umm_var_true}), and to set $x_i$
\false, the path begins at $c_{i_1}$ and goes through vertex $v_{i_1}$ and
$c_{i_2}$.  The chains will connect onto the $a_i$ nodes with odd subscripts 
being $x_i$ ($a_{i_1},a_{i_3},\dots$), and even subscripts representing connections
for $\lnot x_i$ ($a_{i_2}, a_{i_4},\dots$).  The $a_i$ vertices connect by being
shared on one edge of a chain as shown in the example of Figure \ref{fig:umm_example}.
Thus, a \true\ path in the variable does not use the $a_i$ vertex where the chain attaches
and the chain can use that edge (for a \true\ path), but
a \false\ variable path does require the vertex $a_i$ which means the chain will not be able to
use the edge that attaches and will also have a \false\ path setting.
We also note that the nodes $p'_{i_3}$, $p'_{i_6}$, $p'_{i_9}$, etc. are within $\varepsilon$ of 
only the $a_i$ and $v_i$ vertices.

%

\begin{figure}[ht!]
    \begin{center}
        \subfigure[True Variable Gadget]{\label{fig:umm_var_true}\includegraphics[width=.43\textwidth]{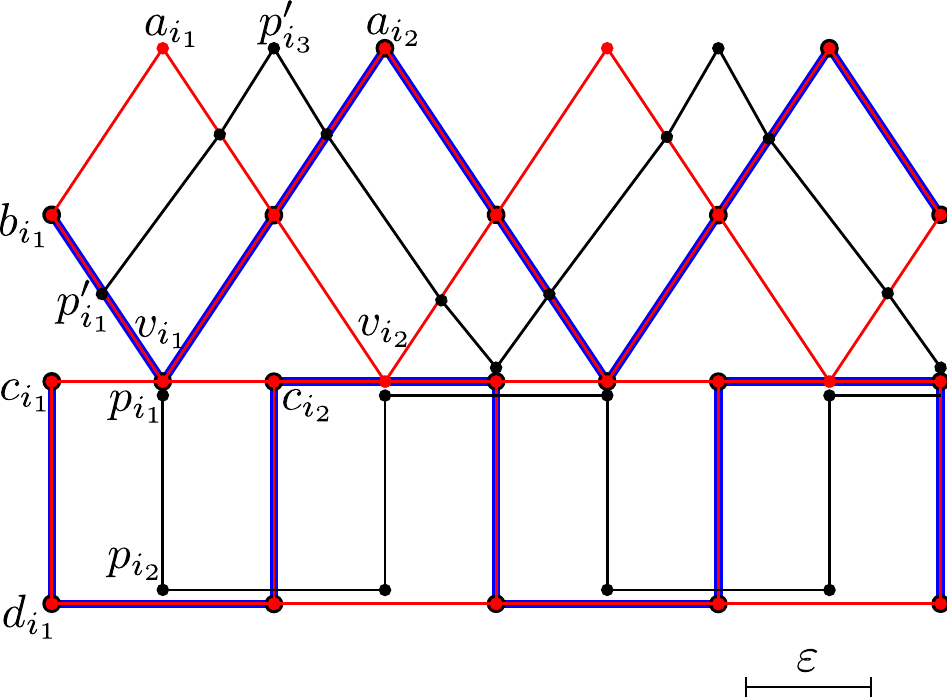}}
        \hspace*{.5cm}
        \subfigure[False Variable Gadget]{\label{fig:umm_var_false}\includegraphics[width=.43\textwidth]{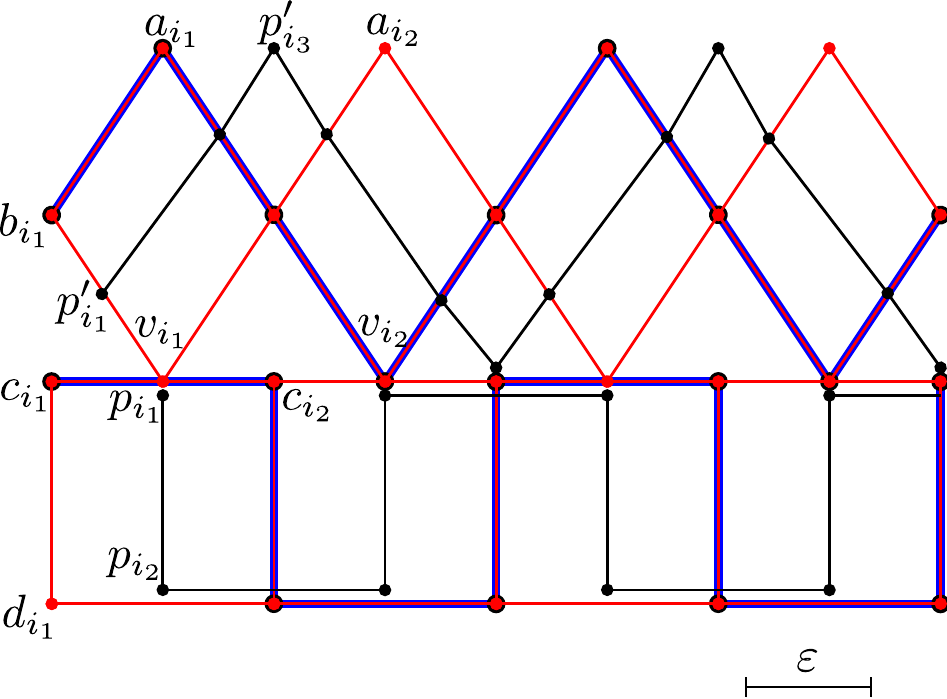}} \\
    \end{center}
    \caption{Variable gadgets with path settings for (a) \true\ (b) \false. A path begins at $c_{i_1}$ and 
        passes through $d_{i_1}$ for \true\ and through $v_{i_1}$ and $c_{i_2}$ for \false.}
    \label{fig:umm_var}
\end{figure}


\subsection{The Clause Gadget}
For the clause gadget, we assume that three `chains' are attached to variable gadgets
and then meet at the junction shown in Figure \ref{fig:umm_clause}. The clause gadget is
planar, but our path along the edges of the graph will no longer be a polygonal curve. 
The \dfd\ is valid since the distances are based on the nodes and not the edges
of the curve, but this would drastically alter the continuous \frechet\ distance.

The vertices in the center are important since many of the edges curve (or turn) in the
space without a vertex.  This is necessary for our reduction, and thus is true
for distance in the space, but does not hold for the \dfd\ based on network
distance \cite{Fan:2011:CGGA}.
We can attach chains at one vertex to the top vertex in the variable gadgets.
These chains then lead to the clause gadget.  See Figure \ref{fig:umm_example}
for an example.

\begin{figure}[!ht]
    \centering
    \includegraphics[width=.6\textwidth,keepaspectratio]{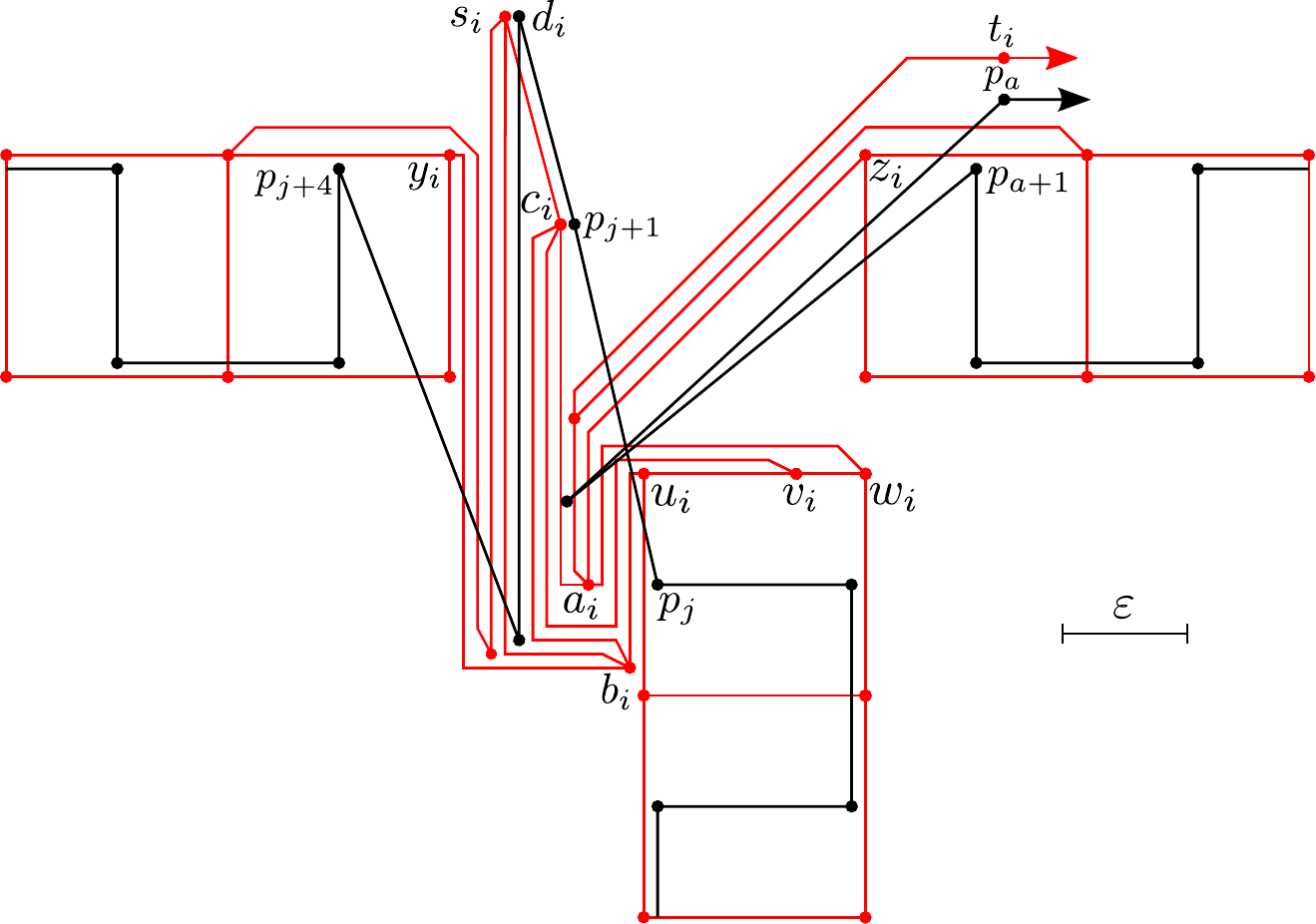}
    \caption{A clause gadget with three chains meeting.}
    \label{fig:umm_clause}
\end{figure}

At a high level, if a path is \true\ in the chain it will not need to use
the last edge.  We will refer to the three variables in our clause as $x_l$, $x_c$, and 
$x_r$ for left, center, and right, respectively.  For $x_l$, if the path is \true, then
it can use the edge to $s_i$.  If the path is \false\ it will end at $y_i$ and there is
only one possible edge.  Similarly, if $x_r$ is \true, it uses the edge to go to
$t_i$, but a \false\ path ends at $z_i$.
In understanding the clause gadget of clause $C_i$, we need focus on the two vertices
$a_i$ and $b_i$. These two act as the Boolean `or' operators which allow only two of the
variables to be \false, and thus requiring at least one to be \true.  Since vertices can
only be used once in a path, each one (of $a_i$ and $b_i$) allow only one edge coming
from a variable to be used.

When $x_l$ is \false, the path ending at $y_i$, it must follow onto $b_i$ and then
go to $s_i$.  Similarly, if $x_r$ has a \false\ path ending at $z_i$, it must follow 
onto $a_i$ and then go to $t_i$.  Notice that $a_i$ and $b_i$ are within $\varepsilon$
of $p_j$.  If $x_c$ has a \false\ path, then it can go from $w_i$ to $a_i$ or go to
$u_i$ and then $b_i$.  This means that if it is \false, it can go to either of the 
center vertices attached to $x_l$ ($b_i$) or $x_r$ ($a_i$).  Thus, only two of them
can be \false.  If $x_c$ is \true, it still needs an edge
to get to $c_i$ without crossing the edges attached to $a_i$ and $b_i$.  If $x_c$
is \true, then a path can go through $u_i$ and then through $v_i$ to $c_i$.
Thus, our clause works as a Boolean `or' of the three variables where at least one
variable must be \true.

\begin{figure}[!ht]
    \centering
    \includegraphics[width=.9\textwidth,keepaspectratio]{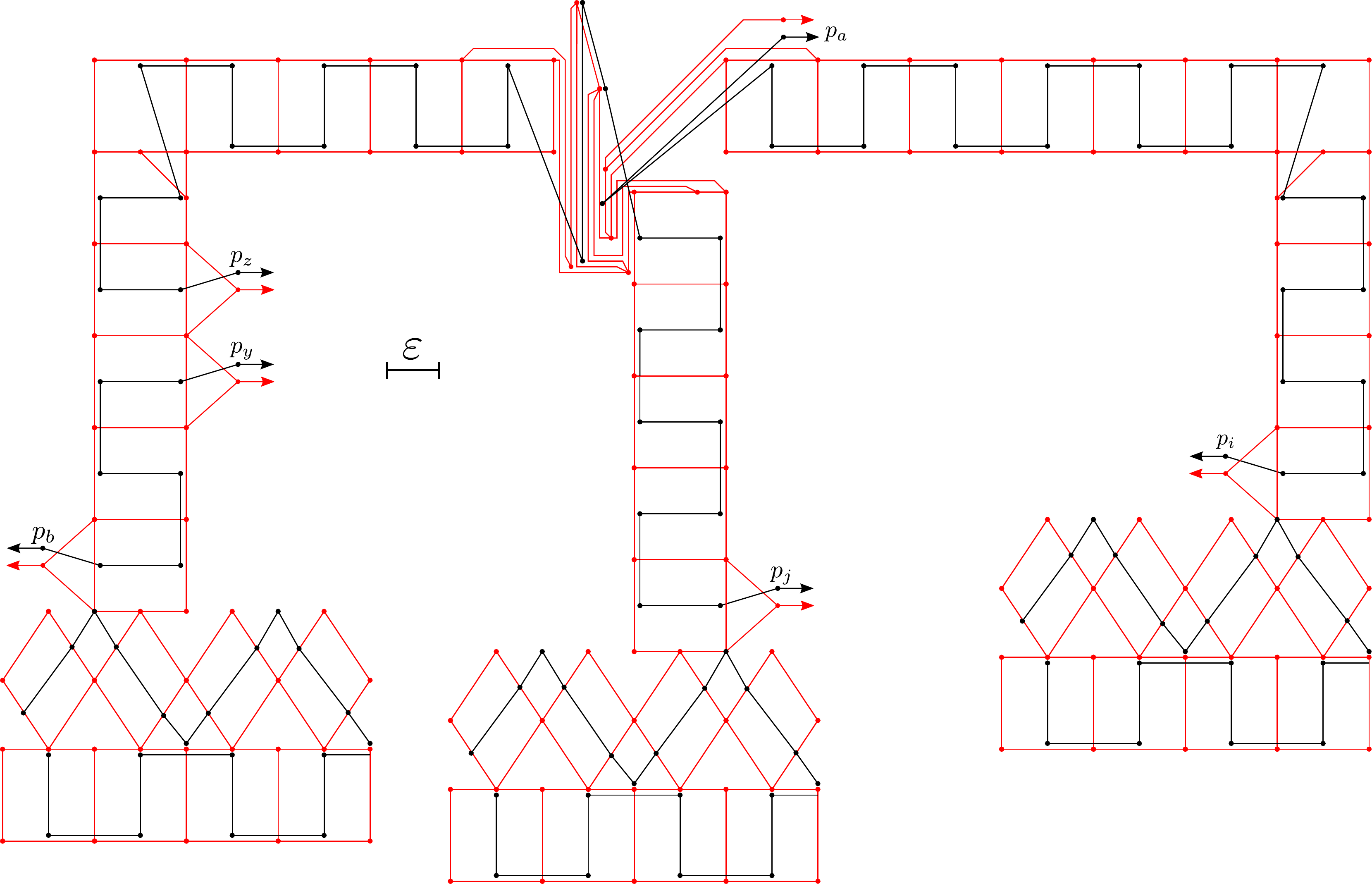}
    \caption{Example UMM clause with three variables $c_i=(\lnot x_1 \lor x_2 \lor \lnot x_3)$ with assignments $x_1=0,$ 
    $x_2=0,$ $x_3=1$.}
    \label{fig:umm_example}
\end{figure}



\subsection{Connecting the Gadgets}

Now we discuss how everything is connected so that $G$ is planar and connected, and that $P$ is
a single continuous polygonal curve. Referring to our clause,
$c_i$ and $s_i$ meet via $d_i$ which allows a path to connect the two variables
without changing their path settings.  Looking at our example (Figure \ref{fig:umm_example})
we see that this leaves six polygonal curve endings (or three segments) in the clause gadget and connecting chains.  
The two exiting nodes ($p_i,p_j$)
from $x_c$ and $x_r$ can simply be connected (or attached to the outside of another nested clause).
Similarly, $p_y$ and $p_z$ can simply be connected if there are no nested clauses under the left leg.
This now means each clause uses a single polygonal curve and has two nodes ($p_a,p_b$) to attach to other 
clauses or the variables along with the associated graph edges.

We can see in the example clause in Figure \ref{fig:umm_example} that the graph and polygonal curve leave
on the outside of the clause.  By design, the middle and right leg connect on the interior.  Thus, in
a planar 3-SAT instance, if there is a nested clause, these lines ($p_i,p_j$) connect to the outside of the nested clause.
Similarly, the points ($p_y,p_z$) connect to the outside of any clause under the left side.  In both cases,
if there is no other clause, the pair of points connect to each other.
If there are multiple clauses under one side, then they are chained together, 
i.e., $p_a$ from one clause connects to $p_b$ from the other.

For some planar 3-SAT instances, it is necessary to attach chains above and below the variable gadget.
We can attach another diamond structure on the other side of our variable gadget and have
clauses on both sides of the variables. The only difference is that it begins with $\lnot x$ in the alternating connections.

We can then define a simple process to connect them.  When saying we are connecting we mean to put
an edge between the two nodes of the polygonal curves and to put an edge between the vertices of the
graph.  We only refer to one of the connections for simplicity.
The order the polygonal chain segments are connected is unimportant as long as the two ends are not within $\varepsilon$ of
each other, and the two graph components can be connected without intersection.

First, connect the variables together. Attach the ending vertex and polygonal node to the 
start of the next variable.  Connect the chain sections and then the diamond structures.
Following, connect the last variable to the closest outside clause variable ($p_a$).  
The clauses are then chained together by outside vertices and the nested clauses link to the inside edges.  
As mentioned, if there are no nested clauses, then connect the pairs ($p_i,p_j$) and ($p_y,p_z$). 


\subsection{The Reduction}

\begin{theorem} \label{thm:ummck}
    The discrete unique map matching (UMM) problem is \npc\ for planar graphs.
\end{theorem}

\begin{proof}
    Given a planar 3-SAT instance $G_{\varphi}=\{V,E\}$ with vertices $V=X \cup C$ such that the vertices represent 
    variables $X=\{x_1,x_2,\dots,x_N\}$ and clauses $C=\{C_1,C_2,...,C_M\}$, and the edges $E=\{e_1,e_2,\dots,e_Z\}$ 
    connect variables to clauses with the degree of each $C_i \in C$ being three.
    Given the planar 3-SAT instance $G_{\varphi}$,  
    we construct a polygonal curve $P$ and a planar graph $H$ using an $\varepsilon > 0$ based on the method described. 
    This construction takes $\BO(|C|+|X|+|E|)$ and is polynomial.
    The sizes of $P$ and $H$ are dependent on $\varepsilon$ and the metric space.  
    In general, for any edge $e_i \in E$ in the space, where $\lVert e_i \rVert$ is the length of the edge, there are 
    $\lceil 2\lVert e_i \rVert/\varepsilon \rceil$ vertices and edges in $H$, and nodes of $P$ used to transfer information
    along that edge.
    
    The planar 3-SAT formula $\varphi$ is satisfiable if and only if there exists a path $Q$ with nodes
    from the vertices in $H$ such that $\dfre(P,Q) \leq \varepsilon$ and each vertex represents a unique node in $Q$.
    
    Given $\varphi$ is satisfiable, then for every clause, there is at least one variable which has a \true\
    value. In our construction this means at least one chain does not need a path through the two points ($a_i,b_i$ for clause $C_i$),
    and thus we can easily find a $Q$ such that $\dfre(P,Q) \leq \varepsilon$.
    
    
    In the other direction, assume there exists a path $Q$ through $V' \subset V$ of $H$ such that $\dfre(P,Q) \leq \varepsilon$.
    There must be at least one \true\ path in a chain at each clause, and since the three chains keep this setting back to
    the variable we know it had this setting at the variable.  Since we also know that the variables
    alternate between a \true\ and \false\ setting, the attached chain has the correct Boolean value associated with the path.  
    Thus, for every variable attached to a clause, it has the correct \true\ or \false\ path setting.  
    Therefore, if $\dfre(P,Q) \leq \varepsilon$, then the current setting of each variable satisfies $\varphi$. 
    
    
    Last, we know the problem is in \textbf{NP}.  Given an instance $I$ we can check whether $\dfre(P,I) \leq \varepsilon$
    in $\BO(|P||I|)$ time via Theorem \ref{thm:dfdtime}.
    \hfill $\square$
\end{proof}

Although our focus is on the \dfd, our reduction actually works to show that the UMM problem is still
\npc\ when based on the continuous \frechet\ distance.  This was simultaneously and independently proven in
\cite{Meulemans:2013:CORR}, but has not been formally published. Since this result is available, we do not rigorously show that
our construction works for the continuous \frechet\ distance here.

\section{Conclusion} \label{sec:conclusion}


In this paper we looked at discrete map matching based on the \dfd\ for the first time, and further defined
some variations based on restricting the problem to unique nodes, 
the number of nodes allowed in the curve, or the number of vertices to choose from.  
We proved that the unique nodes version is \textbf{NP}-complete.
If the number of vertices is restricted, we proved that the problem is \textbf{NP}-complete for regular graphs,
and for planar graphs under the Hausdorff distance.
We proved that finding any path while only limiting the length of the path is polynomial, and  
gave the recurrences for a dynamic programming implementation.
We conclude with a few open questions.

(1) What is the complexity of NMMS based on the discrete and continuous \frechet\ distance?

(2) Are there good approximation algorithms for the optimization versions?

(4) Are there any special cases that are tractable for real-world applications?


\nocite{Wylie:2013:PHD}

\bibliography{map}
\bibliographystyle{abbrv}

\end{document}